\documentclass[12pt]{amsart} 
\usepackage{amsfonts,amsmath,amssymb,amsthm}
\usepackage{underscore}           
\usepackage[usenames]{color}
\usepackage{hyperref}
\usepackage{url}

\newtheorem{theorem}{Theorem}
\newtheorem{claim}[theorem]{Claim}
\newtheorem{corollary}[theorem]{Corollary}
\newtheorem{definition}[theorem]{Definition}
\newtheorem{lemma}[theorem]{Lemma}
\newtheorem{proposition}[theorem]{Proposition}
\newtheorem{remark}[theorem]{Remark}

\newcommand{\scr}[1]{\ensuremath{\mathcal {#1}}}

\newcommand{\eps}{\varepsilon}
\newcommand{\sq}[1]{\ensuremath{\langle#1\rangle}}

\newcommand{\notarrow}{\kern .42em\not\kern -.42em\longrightarrow}
\newcommand{\ket}[1]{\ensuremath{|#1\rangle}}
\newcommand{\bra}[1]{\ensuremath{\langle#1|}}
\renewcommand{\phi}{\varphi}

\newcommand{\nm}[1]{\ensuremath{\Vert #1 \Vert}}

\renewcommand{\a}{\vec a}
\newcommand{\B}{\mathcal B}
\newcommand{\F}{\mathcal F}
\renewcommand{\H}{\mathcal H}
\newcommand{\K}{\mathcal K}
\renewcommand{\L}{\mathcal L}
\newcommand{\M}{\mathcal M}
\newcommand{\m}{\vec m}
\newcommand{\n}{\vec n}
\renewcommand{\O}{\mathcal O}
\newcommand{\R}{\mathbb R}
\newcommand{\s}{\vec\sigma}
\newcommand{\T}{\mathcal T}
\newcommand{\V}{\mathcal V}

\newcommand{\Dim}[1]{\text{Dim}(#1)}
\newcommand{\Tr}[1]{\text{Tr}(#1)}

\newcommand{\noprint}[1]{\relax}

\title[Value and expectation no-go theorems]{Common denominator for
  value and expectation no-go theorems} 
\author{Andreas Blass}
\address{Mathematics Department\\
University of Michigan\\
Ann Arbor, MI 48109--1043, U.S.A.}
\email{ablass@umich.edu}
\author{Yuri Gurevich}
\address{Microsoft Research\\
One Microsoft Way\\
Redmond, WA 98052, U.S.A.}
\email{gurevich@microsoft.com}

\begin{document}
\maketitle

\begin{abstract}
Hidden-variable (HV) theories allege that a quantum state  describes
an ensemble of systems distinguished by the values of hidden
variables. No-go theorems assert that HV theories cannot match the
predictions of quantum theory. The present work started with repairing
flaws in the literature on no-go theorems asserting that HV theories
cannot predict the expectation values of measurements. That literature
gives one an impression that expectation no-go theorems subsume the
time-honored no-go theorems asserting that HV theories cannot predict
the possible values of measurements. But the two approaches speak
about different kinds of measurement. This hinders comparing them to
each other.  Only projection measurements are common to both. Here, we
sharpen the results of both approaches so that only projection
measurements are used.  This allows us to clarify the similarities and
differences between the two approaches. Neither one dominates the
other. 
\end{abstract}

\section{Introduction}
\label{sec:intro}

Hidden-variable theories allege that a state of a quantum system, even
if it is pure and thus contains as much information as quantum
mechanics permits, actually describes an ensemble of systems with
distinct values of some hidden variables. Once the values of these
variables are specified, the system becomes determinate or at least
more determinate than quantum mechanics says.  Thus the randomness in
quantum predictions results, entirely or partially, from the
randomness involved in selecting a member of the ensemble.

No-go theorems assert that, under reasonable assumptions, a hidden-variable interpretation cannot reproduce the predictions of quantum mechanics. In this paper, we examine two species of such theorems, \emph{value} no-go theorems and \emph{expectation} no-go theorems.
The value approach originated in the work of Bell \cite{bell64,bell66} and of Kochen and Specker \cite{ks} in the 1960's. Value no-go theorems establish that, under suitable hypotheses, hidden-variable theories cannot reproduce the predictions of quantum mechanics concerning the possible results of the measurements of observables.

The expectation approach was developed in the last decade by Spekkens \cite{spekkens} and by Ferrie, Emerson, and Morris \cite{ferrieA,ferrieB,ferrieC}, with \cite{ferrieC} giving the sharpest result.  In this approach, the discrepancy between hidden-variable theories and quantum mechanics appears in the predictions of the expectation values of the measurements of effects, i.e.\ the elements of POVMs, positive operator-valued measures.  There is no need to consider the actual values obtained by measurements or the probability distributions over these values.

In both cases, measurements are associated to Hermitian operators, but
they are different sorts of measurements.  In the value approach,
Hermitian operators serve as observables, and measuring one of them
produces a number in its spectrum.  In the expectation approach,
certain Hermitian operators serve as effects, and measuring one of
them produces 0 or 1, even if the spectrum consists entirely of other points.
The only Hermitian operators for which these two uses coincide are
projections.

We sharpen the results of both approaches so that only projection
measurements are used. Regarding the expectation approach, we
substantially weaken the hypotheses. We do not need arbitrary effects,
but only rank-1 projections. Accordingly, we need convex-linearity
only for the hidden-variable picture of states, not for that of
effects. Regarding the value approach, it turns out that rank-1
projections are sufficient in the finite dimensional case but not in
general. Finally, using a successful hidden-variable theory of John
Bell for a single qubit, we demonstrate that the expectation approach does not subsume
the value approach.

\section{Expectation No-Go Theorem}
\label{sec:exp}

\begin{definition}\label{def:exp}\rm
An \emph{expectation representation} for quantum systems described by a Hilbert space $\H$ is a triple $(\Lambda,\mu,F)$ where
\begin{itemize}
\item $\Lambda$ is a measurable space,
\item $\mu$ is a convex-linear map assigning to each density operator $\rho$ on $\H$ a probability measure $\mu(\rho)$ on $\Lambda$, and
\item $F$ is a map assigning to each rank-1 projection $E$ in $\H$ a measurable function $F(E)$ from $\Lambda$ to the real interval $[0,1]$.
\end{itemize}
It is required that for all density matrices $\rho$ and all rank-1 projections $E$
\begin{equation}\label{eq1}
\Tr{\rho\cdot E}=\int_\Lambda F(E)\,d\mu(\rho)
\end{equation}
\end{definition}

The convex linearity of $\mu$ means that $\mu(a_1\rho_1 + a_2\rho_2) = a_1\mu(\rho_1) + a_2\mu(\rho_2)$ whenever $a_1,a_2$ are nonnegative real numbers with sum 1.

The definition of expectation representation is similar to  Ferrie-Morris-Emerson's definition of the probability representation \cite{ferrieC} except that (i)~the domain of $F$ contains only rank-1 projections, rather than arbitrary effects, and (ii)~we do not (and cannot) require that $F$ be convex-linear.

Intuitively an expectation representation $(\Lambda,\mu,F)$ attempts to predict the expectation value of any rank-1 projection $E$ in a given mixed state $\rho$. The hidden variables are combined into one variable ranging over $\Lambda$. Further, $\mu(\rho)$ is the probability measure on $\Lambda$ determined by $\rho$, and $(F(E))(\lambda)$ is the probability of determining the effect $E$ at the subensemble of $\rho$ determined by $\lambda$. The left side of \eqref{eq1} is the expectation of $E$ in state $\rho$ predicted by quantum mechanics and the right side is the expectation of $F(E)$ in the ensemble described by $\mu(\rho)$.

But why is $\mu$ supposed to be convex linear? Well, mixed states have
physical meaning and so it is desirable that $\mu$ be defined on mixed
states as well. If you are a hidden-variable theorist, it is most
natural for you to think of a mixed state as a classical probabilistic
combination of the component states. This leads you to the convex
linearity of $\mu$. For example, if $\rho = \sum_{i=1}^k p_i\rho_i$
where $p_i$'s are nonnegative reals and $\sum p_i = 1$ then, by the
rules of probability theory, $(\mu(\rho))(S) = \sum p_i(\mu(\rho_i))(S)$
for any measurable $S\subseteq\Lambda$. Note, however, that you cannot
start with any wild probability distribution $\mu$ on pure states and
then extend it to mixed states by convex linearity. There is an
important constraint on $\mu$. The same mixed state $\rho$ may have
different representations as a convex combination of pure states; all
such representations must lead to the same probability measure
$\mu(\rho)$.

\begin{theorem}[First Bootstrapping Theorem]\label{thm:boot1}
Let $\H$ be a closed subspace of a Hilbert space $\H'$.  From any expectation representation for quantum systems described by $\H'$, one can directly construct such a representation for systems described by $\H$.
\end{theorem}

\begin{proof}
We construct an expectation representation $(\Lambda,\mu,F)$ for quantum systems described by $\H$ from any such representation $(\Lambda',\mu',F')$ for the larger Hilbert space $\H'$.  To begin, we set $\Lambda=\Lambda'$.

To define $\mu$ and $F$, we use the inclusion map $i:\H\to\H'$,
sending each element of $\H$ to itself considered as an element of
$\H'$, and we use its adjoint $p:\H'\to\H$, which is the orthogonal
projection of $\H'$ onto $\H$.  Any density operator $\rho$ over $\H$,
gives rise to a density operator $\bar\rho=i\circ\rho\circ p$ over
$\H'$.  Note that this expansion is very natural: If $\rho$
corresponds to a pure state $\ket\psi\in\H$, i.e., if
$\rho=\ket\psi\bra\psi$, then $\bar\rho$ corresponds to the same
$\ket\psi\in\H'$.  If, on the other hand, $\rho$ is a mixture of
states $\rho_i$, then $\bar\rho$ is the mixture, with the same
coefficients, of the $\overline{\rho_i}$.  Define
$\mu(\rho)=\mu'(\bar\rho)$.

The definition of $F$ is similar. For any rank-1 projection $E$ in $\H$, $\bar E=i\circ E\circ p$ is a rank-1 projection in $\H'$, and so we define $F(E)=F'(\bar E)$. If $E$ projects to the one-dimensional subspace spanned by $\ket\psi\in\H$, then $\bar E$ projects to the same subspace, now considered as a subspace of $\H'$.

This completes the definition of $\Lambda$, $\mu$ , and $F$.
Most of
the requirements in Definition~\ref{def:exp} are trivial to verify.
For the last requirement, the agreement between the expectation
computed as a trace in quantum mechanics and the expectation computed
as an integral in the expectation representation, it is useful to
notice first that $p\circ i$ is the identity operator on $\H$.  We can then
compute, for any density operator $\rho$ and any rank-1 projection
$E$ on $\H$,
\begin{align*}
  \int_\Lambda F(E)\,d\mu(\rho)&=\int_\Lambda F'(\bar E)\,d\mu'(\bar\rho)
  =\Tr{\bar\rho\bar E}
  =\Tr{i\circ\rho\circ p\circ i\circ E\circ p}\\
  &=\Tr{i\circ\rho\circ E\circ p}
  =\Tr{\rho\circ E\circ p\circ i} = \Tr{\rho\circ E},
\end{align*}
as required.
\end{proof}

\begin{theorem}[Expectation no-go theorem]\label{thm:exp}
  If the dimension of the Hilbert space $\H$ is at least 2 then there is no expectation representation for quantum systems described by $\H$.
\end{theorem}

We cannot expect any sort of no-go result in lower dimensions, because quantum theory in Hilbert spaces of dimensions 0 and 1 is trivial and therefore classical.
By the First Bootstrapping Theorem, it suffices to prove Theorem~\ref{thm:exp} just in the case $\Dim{\H}=2$.
But we find Ferrie-Morris-Emerson's proof that works directly for
all dimensions \cite{ferrieC} instructive, and we adjust it to prove
Theorem~\ref{thm:exp}.
The adjustment involves adding some details and observing that a
drastically reduced domain of $F$ suffices. The adjustment also
involves making a little correction. Ferrie et al.\ quoted an
erroneous
result of Bugajski \cite{bugajski} which needs some additional
hypotheses to become correct. Fortunately for Ferrie et al., those
hypotheses hold in their situation.

\begin{proof}
The proof involves several normed vector spaces.
\begin{itemize}
\item $\B$ is the real Banach space of bounded self-adjoint operators
  $\H\to\H$ with norm\\ $\nm A=\sup\{\nm{Ax}:x\in \H,\,\nm x=1\}$.
\item $\F$ is the real vector space of bounded, measurable, real-valued functions on $\Lambda$ with norm\\ $\nm  f = \sup \{|f(\lambda)|: \lambda\in\Lambda\}$.
\item $\M$ is the real vector space of bounded, signed, real-valued measures on $\Lambda$ with the total variation norm $\nm\mu = \mu_+(\Lambda) + \mu_-(\Lambda)$ where $\mu=\mu_+-\mu_-$ and $\mu_+$ and $\mu_-$ are positive measures with disjoint supports.
\item $\T$ is the vector subspace of $\B$ consisting of the trace-class operators. These are the operators $A$ whose spectrum consists of real eigenvalues $\alpha_i$ such that the sum $\sum_i|\alpha_i|$ is finite; eigenvalues with multiplicity $>1$ are repeated in this list, and the continuous spectrum is empty or $\{0\}$. The sum $\sum_i|\alpha_i|$ serves as the norm of $A$ in $\T$. The sum $\sum_i\alpha_i$ of eigenvalues themselves (rather than their absolute values) is the trace of $A$. Note that density operators are positive trace-class operators of trace 1.
\end{itemize}

In the rest of the proof, by default, operators, transformations and
functionals are bounded  and of course linear. Suppose, toward a contradiction, that we have an expectation representation $(\Lambda,\mu,F)$ for some $\H$ with $\Dim{\H}\ge2$.

\begin{lemma}\label{lem:e1}
$\mu$ can be extended in a unique way to a transformation, also denoted $\mu$, from all of $\T$ into $\M$.
\end{lemma}

\begin{proof}[Proof of Lemma~\ref{lem:e1}]
Every $A\in\T$ can be written as a linear combination of two density operators. Indeed, if $\nm A>0$ and $A$ is positive then $\Tr{A}=\nm A$ and $A={\nm A}\rho$ where $\rho=\frac{A}{\nm A}$. In general, it suffices to represent $A$ as the difference $B-C$ of positive trace-class operators. Choose $A_+$ (resp.\ $A_-$) to have the same positive (resp.\ negative) eigenvalues and corresponding eigenspaces as $A$ and be identically zero on all the eigenspaces corresponding to the remaining eigenvalues.  The desired $B = A_+$ and $C = - A_-$.

If $A$ is a linear combination $b\rho+c\sigma$ of two density
operators, define $\mu(A)=b\mu(\rho)+c\mu(\sigma)$.
Using the convex linearity of $\mu$ on the density operators, it is easy to check that if $A$ has another such representation $b'\rho'+c'\sigma'$ then
$b\mu(\rho) + c\mu(\sigma) = b'\mu(\rho') + c'\mu(\sigma')$ which means that $\mu(A)$ is well-defined.

The uniqueness of the extension is obvious. It remains to check that
the extended $\mu$ is bounded. In fact, we show more, namely that
$\nm{\mu(A)} \le1$ if $\nm A \le1$. 
So let $A\in\T$ and $\nm A \le1$. As we saw above, there are
positive trace-class operators $B,C$ such that $A = B-C$. Then $A =
{\nm B}\rho - {\nm C}\sigma$ for some density operators $\rho, \sigma$
where $b,c\ge0$ and $b+c={\nm A}\le1$. Now, $\mu(\rho)$ and
$\mu(\sigma)$ are measures with norm 1.  So $\nm{\mu(A)} \le b{\mu(\rho)}
+ c{\mu(\sigma)}\le b+c \le 1$.
\end{proof}

Let $\M'$ be the space of the functionals $\M\to\R$ where $\R$ is the
set of real numbers. Similarly let $\T'$ be the space of the
functionals $\T\to\R$. $\mu$ gives rise to a dual transformation
$\mu': \M'\to\T'$ that sends any $h\in\M'$ to $\mu'(h) = h\circ\mu$ so
that 
\begin{equation}\label{eq2}
\mu'(h)(A) = h(\mu(A))\quad \text{for all $h\in\M'$ and all $A\in\T$}.
\end{equation}
Every measurable function $f\in\F$ induces a functional $\bar f \in\M'$ by integration: $\bar f(\mu)=\int_\Lambda f\,d\mu$. This gives rise to a transformation $\nu: \F\to\T'$ that sends every $f$ to $\mu'(\bar f)$. Specifying $h$ to $\bar f$ in Equation~\ref{eq2} gives
\begin{equation}\label{eq3}
(\nu f)(A) = \int_\Lambda f\,d\mu(A)\quad
\text{for all $f\in\F$ and all $A\in\T$}.
\end{equation}
Here and below we omit the parentheses around the argument of $\nu$.

\begin{lemma}\label{lem:e2}
For every $f\in\F$, there is a unique $B\in\B$ with $(\nu f)(\rho) =
\Tr{B\cdot\rho}$ for all density operators $\rho$. 
\end{lemma}

\begin{proof}[Proof of Lemma~\ref{lem:e2}]
Every $B\in\B$ induces a functional $\bar B\in\T'$ by $\bar
B(A)=\Tr{B\cdot A}$. Here $\Tr{B\cdot\rho}$ is well-defined because
the product of a bounded operator and a trace-class operator is again
in the trace class \cite[Lemma~3, p.\ 38]{schatten}.

The map $B\mapsto\bar B$ is an isometric isomorphism between $\B$ and $\T'$ \cite[Theorem~2, p.~47]{schatten}.
So, for every $X\in\T'$, there is a unique $B_X\in\B$ such that $X(A) = \Tr{B_X\cdot A}$ for all $A\in\T$. Furthermore, there is a unique $B_X\in\B$ such that $X(\rho) = \Tr{B_X\cdot\rho}$ for all density operators $\rho$.  This is because, as we showed above, the linear span of the density matrices is the whole space $\T$.  The lemma follows because every $\nu f$ belongs to $\T'$.
\end{proof}

For any $f\in\F$, the unique operator $B$ with $(\nu f)(\rho) =
\Tr{B\cdot\rho}$ for all $\rho$ will be denoted $[\nu f]$. 

\begin{lemma}\label{lem:e3}
$[\nu F(E)] = E$ for every rank-1 projection $E$, and $[\nu1]=I$ where 1 is the constant function with value 1 and $I$ is the unit matrix.
\end{lemma}

\begin{proof}[Proof of Lemma~\ref{lem:e3}]
Lemma~\ref{lem:e2} and equation~\eqref{eq3} give

\begin{equation}\label{eq4}
\Tr{[\nu f]\cdot\rho}=\int_\Lambda f\,d\mu(\rho)\quad
\text{for every density operator }\rho.
\end{equation}

Equations~\eqref{eq1} and \eqref{eq4} imply
\begin{equation}\label{eq5}
\Tr{\rho\cdot E} = \int_\Lambda F(E)\,d\mu(\rho)
                 = \Tr{[\nu F(E)]\cdot\rho}
\end{equation}%
for every density operator $\rho$.

The right sides of Equations~\eqref{eq1} and \eqref{eq4} coincide if
we specify  $f$ to $F(E)$. 
Therefore their left sides are equal.

\begin{equation*}
\Tr{E\rho} = \Tr{[\nu F(E)]\rho}
\end{equation*}

We now invoke the last clause in Definition~\ref{def:exp} to find that, for
all rank-1 projections $E$ and all density matrices $\rho$,
\[
\Tr{E\rho}=\int_\Lambda F(E)\,d\mu(\rho) = \Tr{[\nu F(E)]\rho)}.
\]
But this is, as we saw in the proof of Lemma~\ref{lem:e2}, enough to show that
$[\nu F(E)]=E$.

By Lemma~\ref{lem:e2}, we see that $\mu'(1)$ is the unique operator that satisfies, for all $\rho$,
\[
\Tr{[\nu1]\rho}=\int_\Lambda\,d\mu(\rho) = (\mu(\rho))(\Lambda)=1=\Tr\rho
=\Tr{I\rho},
\]
where the third equality comes from the fact that $\mu$ maps density matrices to  probability measures. Thus, $[\nu1]=I$.
\end{proof}

\begin{lemma}\label{lem:e4}
  For any two rank-1 projections $A,B$ of $\H$, there exists an operator $H\in\B$ such that all four of $H$, $A-H$, $B-H$, and $I-A-B+H$ are positive operators.
\end{lemma}

\begin{proof} [Proof of Lemma~\ref{lem:e4}]
Recall that an operator $A$ is said to be
positive if $\bra\psi A\ket\psi\geq0$ for all $\ket\psi\in\H$ and
that $A\leq B$ means that $B-A$ is positive.  A function $f\in\F$ is \emph{nonnegative} if  $f(\lambda)\geq0$ for all $\lambda\in\Lambda$.

\begin{claim} \label{cla:e1} If $f\in\scr F$ is nonnegative then $[\nu f]$ is a
  positive operator.  Therefore, if $f\leq g$ pointwise in $\F$ then
  $[\nu f]\le [\nu g]$ in $\B$.
\end{claim}

\begin{proof}[Proof of Claim~\ref{cla:e1}]
  The second assertion follows immediately from the first applied to
  $g-f$, because $\nu$ is linear.  To prove the first assertion,
  suppose $f\in\scr F$ is nonnegative, and let $\ket\psi$ be any
  vector in $\H$. The conclusion we want to deduce, $\bra\psi
  [\nu f]\ket\psi\geq0$, is obvious if $\ket\psi=0$, so we may assume
  that \ket\psi\ is a non-zero vector.  Normalizing it, we may assume
  further that its length is 1.  Then $\ket\psi\bra\psi$ is a density operator
  and therefore $\mu(\ket\psi\bra\psi)$ is a measure. Using
  equation~\eqref{eq5}, we compute
\[
\bra\psi [\nu f]\ket\psi=\Tr{[\nu f]\ket\psi\bra\psi}
=\int_\Lambda f\,d\mu(\ket\psi\bra\psi)\geq0,
\]
where we have used that both the measure $\mu(\ket\psi\bra\psi)$ and
the integrand $f$ are nonnegative.
\end{proof}

Let $\F_{[0,1]}$ be the subset of $\F$ comprising the functions all of whose values are in the interval $[0,1]$.

\begin{claim}\label{cla:e2}
  For any $f,g\in\F_{[0,1]}$ there exists $h\in\F_{[0,1]}$ such that all four of $h$, $f-h$, $g-h$, and $1-f-g+h$ are nonnegative.
\end{claim}

\begin{proof}[Proof of Claim~\ref{cla:e2}]
  Define $h(\lambda)=\min\{f(\lambda),g(\lambda)\}$ for all
  $\lambda\in\Lambda$.  Then the first three of the assertions in the
  lemma are obvious, and the fourth becomes obvious if we observe that 
  $f+g-h=\max\{f,g\}\leq 1$.
\end{proof}

Now we are ready to complete the proof of Lemma~\ref{lem:e4}.
Apply Claim~\ref{cla:e2} 
with $f=F(A)$ and $g=F(B)$,
let $h$ be the function given by the lemma, and let $H=[\nu(h)]$.  
The  nonnegativity of $h$, $f-h$, $g-h$, and $1-f-g+h$ implies,  by 
Claim~\ref{cla:e1}, 
the positivity of $[\nu(h)]=H$, $[\nu(F(A)-h)]=A-H$,
$[\nu(F(B)-h)]=B-H$, 
and $[\nu(1-F(A)-F(B)+h)]=I-A-B+H$, where we have also used the
linearity of $\nu$, the fact that $[\nu(1)]=I$, and the formula
$[\nu(F(A))]=A$ for all $A$ in the domain of $F$.
\end{proof}

Now we are ready to prove Theorem~\ref{thm:exp}.
Let us apply Lemma~\ref{lem:e3} to two specific rank-1 projections.  Fix
two orthonormal vectors \ket0 and \ket1. (This is where we use that
$\H$ has dimension at least 2.)  Let $\ket+=(\ket0+\ket1)/\sqrt2$.
We use the projections $A=\ket0\bra0$ and $B=\ket+\bra+$ to the
subspaces spanned by $\ket0$ and $\ket+$.  Let $H$ be as in
Lemma~\ref{lem:e4} for these projections $A$ and $B$.

From the positivity of $H$ and of $A-H$, we get that
$0\leq\bra1H\ket1$ and that
\[
0\leq\bra1(A-H)\ket1=\bra1A\ket1-\bra1H\ket1=-\bra1H\ket1,
\]
where we have used that \ket1, being orthogonal to \ket0, is
annihilated by $A$.  Combining the two inequalities, we infer that
$\bra1H\ket1=0$ and therefore, since $H$ is positive, $H\ket1=0$.
Similarly, using the orthogonal vectors $\ket+$ and
$\ket-=\ket0-\ket1)/\sqrt2$ in place of \ket0 and \ket1, we obtain
$H\ket-=0$.  So, being linear, $H$ is identically zero on the subspace
of $\H$ spanned by \ket1 and \ket-; note that \ket0 is in this
subspace, so we have $H\ket0=0$.

Now we use the positivity of $I-A-B+H$.  Since $H\ket0=0$, we
can compute
\[
0\leq\bra0(I-A-B+H)\ket0=\sq{0|0}-\bra0A\ket0-\bra0B\ket0=
1-1-\frac1{\sqrt2}=\frac{-1}{\sqrt2}.
\]
This contradiction completes the proof of the theorem.
\end{proof}

\begin{remark}[Symmetry or the lack of thereof]\rm
In view of the idea of symmetry or even-handedness suggested by
Spekkens \cite{spekkens}, one might ask whether there is a dual
version of Theorem~\ref{thm:exp}, that is, a version that requires
convex-linearity for effects but looks only at pure states and does not
require any convex-linearity for states.
The answer is no; with such requirements there is a trivial example of a successful hidden-variable theory, regardless of the dimension of the Hilbert space.  The theory can be concisely described as taking the quantum state itself as the ``hidden'' variable.  In more detail, let $\Lambda$ be the set of all pure states.  Let $\mu$ assign to each operator $\ket\psi\bra\psi$ the probability measure on $\Lambda$ concentrated at the point $\lambda_{\ket\psi}$ that corresponds to the vector $\ket\psi$.  Let $F$ assign to each effect $E$ the function on $\Lambda$ defined by
\[
F(E)(\lambda_{\ket\psi})=\bra\psi E\ket\psi.
\]
We have trivially arranged for this to give the correct expectation
for any effect $E$ and any pure state \ket\psi.  The formula for
$F(E)$ is clearly convex-linear (in fact, linear) as a function of
$E$.  Of course, $\mu$ cannot be extended convex-linearly to mixed
states, so that Theorem~\ref{thm:exp} does not apply.
\end{remark}

\section{Value No-Go Theorems}
\label{sec:val}
Value no-go theorems assert that hidden-variable theories cannot even produce the correct outcomes for individual measurements, let alone the correct probabilities or expectation values.  Such theorems considerably predated the expectation no-go theorems considered in the preceding section.  Value no-go theorems were first established by Bell \cite{bell64,bell66} and then by Kochen and Specker \cite{ks}; we shall also refer to the user-friendly exposition given by Mermin \cite{mermin}.  To formulate value no-go theorems, one must specify what ``correct outcomes for individual measurements'' means.

\begin{definition} \label{valmap}
Let $\H$ be a Hilbert space, and let $\O$ be a set of observables, i.e., self-adjoint operators on $\H$. A \emph{valuation} for $\O$ in $\H$ is a function $v$ assigning to each observable $A\in\O$ a number $v(A)$ in the spectrum of $A$, in such a way that $(v(A_1),\dots,v(A_n))$ is in the joint spectrum $\sigma(A_1,\dots,A_n)$ of $(A_1,\dots,A_n)$
whenever $A_1,\dots,A_n$ are pairwise commuting.
\end{definition}

The intention behind this definition is that, in a hidden-variable theory, a quantum state represents an ensemble of individual systems, each of which has definite values for observables. That is, each individual system has a valuation associated to it, describing what values would be obtained if we were to measure observable properties of the system.  A believer in such a hidden-variable theory would expect a valuation for the set of all self-adjoint operators on $\H$, unless there were superselection rules rendering some such operators unobservable.

Before we proceed, we recall the notion of joint spectra \cite[Section~6.5]{spectral}.

\begin{definition}
The \emph{joint spectrum} $\sigma(A_1,\dots,A_n)$ of pairwise
commuting, self-adjoint operators $A_1,\dots,A_n$ on a Hilbert space
$\H$ is a subset of $\R^n$. If $A_1,\dots,A_n$ are simultaneously
diagonalizable then $(\lambda_1,\dots,\lambda_n)
\in\sigma(A_1,\dots,A_n)$ iff there is a non-zero vector $\ket\psi$
with $A_i\ket\psi=\lambda_i\ket\psi$ for $i=1,\dots,n$. In general,
$(\lambda_1,\dots,\lambda_n) \in\sigma(A_1,\dots,A_n)$ iff for every
$\eps>0$ there is a unit vector $\ket\psi\in\H$ with
$\nm{A_i\ket\psi-\lambda_i\ket\psi}<\eps$ for $i=1,\dots,n$.
\end{definition}

\begin{proposition}\label{pro:jspec}\
For any continuous function $f:\R^n\to\R$,\\ $f(A_1,\dots,A_n)=0$ if and only if $f$ vanishes identically on $\sigma(A_1,\dots,A_n)$.
\end{proposition}

The proposition is implicit in the statement, on page~155 of
\cite{spectral}, that ``most of Section~1, Subsection~4, about
functions of one operator,''  
can be repeated in the context of
several commuting operators.  We give a detailed proof of the
proposition in \cite[\S4.1]{G228}. 

\begin{theorem}[\cite{bell66,ks,mermin}]\label{thm:dim3}
If $\Dim{\H}=3$ then there is a finite set $\O$ of rank~1 projections for which no valuation exists.
\end{theorem}

The proof of Theorem~\ref{thm:dim3} can be derived from the work of
Bell \cite[Section~5]{bell66}, and we do that explicitly in
\cite[\S4.3]{G228}. The construction given by Kochen and Specker
\cite{ks}  provides the desired $\O$ more directly.  The proof of
Theorem~1 in \cite{ks} uses a Boolean algebra generated by a finite
set of one-dimensional subspaces of $\H$, and it shows that the
projections to those subspaces constitute an $\O$ of the required
sort.  Mermin's elegant exposition \cite[Section~IV]{mermin} deals
instead with squares $S_i^2$ of certain spin-components of a spin-1
particle, but these are projections to 2-dimensional subspaces of
$\H$, and the complementary rank-1 projections $I-S_i^2$ serve as the
desired $\O$.

\begin{theorem}[Second Bootstrapping Theorem]\label{thm:boot2}
Suppose $\H\subseteq\H'$ are finite-dimensional Hilbert spaces.  Suppose further that $\O$ is a finite set of rank-1 projections of $\H$ for which no valuation exists.  Then there is a finite set $\O'$ of rank-1 projections of $\H'$ for which no valuation exists.
\end{theorem}

This is our second bootstrapping theorem. Intuitively, such dimension bootstrapping results are to be expected. If hidden-variable theories could explain the behavior of quantum systems described by the larger Hilbert space, say $\H'$, then they could also provide an explanation for systems described by the subspace $\H$.  The latter systems are, after all, just a special case of the former, consisting of the pure states that happen to lie in $\H$ or mixtures of such states. But often no-go theorems give much more information than just the impossibility of matching the predictions of quantum-mechanics with a hidden-variable theory.  They establish that hidden-variable theories must fail in very specific ways. It is not so obvious that these specific sorts of failures, once established for a Hilbert space $\H$, necessarily also apply to its
superspaces $\H'$.

\begin{proof}
  Clearly, if two Hilbert spaces are isomorphic and if one of them has
  a finite set $\O$ of rank-1 projections with no valuation, then
  the other also has such a set.  It suffices to conjugate the
  projections in $\O$ by any isomorphism between the two
  spaces. Thus, the existence of such a set $\O$ depends only on the
  dimension of the Hilbert space, not on the specific space.

  Proceeding by induction on the dimension of $\H'$, we see that
  it suffices to prove the theorem in the case where $\dim(\H')=\dim(\H)+1$.  Given such $\H$ and $\H'$, let \ket\psi\
  be any unit vector in $\H'$, and observe that its orthogonal
  complement, $\ket\psi^\bot$, is a subspace of $\H'$ of the same
  dimension as $\H$ and thus isomorphic to $\H$.  By the induction
  hypothesis, this subspace $\ket\psi^\bot$ has a finite set $\O$ of
  rank-1 projections for which no valuation exists.  Each element of
  $\O$ can be regarded as a rank-1 projection of $\H'$; indeed,
  if the projection was given by $\ket\phi\bra\phi$ in
  $\ket\psi^\bot$, then we can just interpret the same formula
  $\ket\phi\bra\phi$ in $\H'$, using the same unit vector
  $\ket\phi\in\ket\psi^\bot$

  Let $\O_1$ consist of all the projections from $\O$,
  interpreted as projections of $\H'$, together with one
  additional rank-1 projection, namely $\ket{\psi}\bra{\psi}$.  What
  can a valuation $v$ for $\O_1$ look like? It must send
  $\ket{\psi}\bra{\psi}$ to one of its eigenvalues, 0 or 1.

Suppose first that $v(\ket{\psi}\bra{\psi})=0$. Then, using the fact
that $\ket{\psi}\bra{\psi}$ commutes with all the other elements of
$\O_1$, we easily compute that what $v$ does to those other
elements amounts to a valuation for $\O$.  But $\O$ was chosen so
that it has no valuation, and so we cannot have
$v(\ket{\psi}\bra{\psi})=0$.  Therefore $v(\ket{\psi}\bra{\psi})=1$.  (It
follows that $v$ maps the projections associated to all the other
elements of $\O'$ to zero, but we shall not need this fact.)

We have thus shown that any valuation for the finite set $\O_1$
must send $\ket\psi\bra\psi$ to 1.  Repeat the argument for another
unit vector $\ket{\psi'}$ that is orthogonal to \ket\psi.  There is a
finite set $\O_2$ of rank-1 projections such that any valuation
for $\O_2$ must send \ket{\psi'}\bra{\psi'} to 1.  No valuation
can send both \ket\psi\bra\psi\ and \ket{\psi'}\bra{\psi'} to 1,
because their joint spectrum consists of only $(1,0)$ and $(0,1)$.
Therefore, there can be no valuation for the union $\O_1\cup\O_2$, which thus serves as the $\O'$ required by the theorem.
\end{proof}

\begin{theorem}[Value no-go theorem]\label{thm:val}
Suppose that the dimension of the Hilbert space is at least 3.
\begin{enumerate}
\item There is a finite set $\O$ of projections for which no valuation exists.
\item If the dimension is finite then there is a finite set $\O$
  of rank~1 projections for which no valuation exists.
\end{enumerate}
\end{theorem}

The desired finite sets of projections are constructed explicitly in the proof. The finiteness assumption in part (2) of the theorem cannot be omitted. If $\Dim{\H}$ is infinite, then the set $\O$ of all finite-rank projections admits a valuation, namely the constant zero function.  This works because the definition of ``valuation'' imposes constraints on only finitely many observables at a time.

\begin{proof}
When the dimension of $\H$ is greater than 3, but still finite, we use our Second Bootstrapping Theorem.  Notice that, if one merely wants a no-go theorem saying that some $\O$ has no valuation, then this bootstrapping is easy, as noted in \cite{bell64,ks,mermin}.  Work is needed only to get all the operators in $\O$ to be rank~1 projections.

It remains to treat the case of infinite-dimensional $\H$. Let $\K$ and $\L$ be Hilbert spaces, with $\dim(\K)=3$ and $\dim(\L)=\dim(\H)$.  Note that then their tensor product $\K\otimes\L$ has the same dimension as $\H$, so it can be identified with $\H$.

Let $\O$ be as in Theorem~\ref{thm:dim3} for the 3-dimensional $\K$.    Let
$\O'=\{P\otimes I_{\L}:P\in\O\}$, where $I_{\L}$ is the identity operator on $\L$.  Then $\O'$ is a set of infinite-rank projections of $\K\otimes\L=\H$, having the same algebraic structure as $\O$.  It follows that there is no valuation for $\O'$.
\end{proof}

Let's say that a projection $A$ on Hilbert space $\H$ is a \emph{rank-$n$ projection modulo identity} if either $A$ is of rank $n$ or else $\H$ splits into a tensor product $\K\otimes\L$ such that $\K$ is finite-dimensional and $A$ has the form $P\otimes I_{\L}$ where $P$ is of rank $n$ and $I_{\L}$ is the identity operator on $\L$. The proof of Theorem~\ref{thm:val} gives us the following corollary.

\begin{corollary}
If the dimension of the Hilbert space is at least 3 then there is a finite set of rank-1 projections modulo identity for which no valuation exists.
\end{corollary}

\section{One successful hidden-variable theory}
\label{sec:bell}

By reducing both species of no-go theorems to projection measurement,
where measurement as observable and measurement as effect coincide, we
made it easier to see similarities and differences.  No, the
expectation no-go theorem does not imply the value no-go theorem. But
the task of proving this claim formally, say for a given dimension
$d=\Dim{\H}$, is rather thankless. You have to construct a
counter-factual physical world where the expectation no-go theorem
holds but the value no-go theorem fails. There is, however, one
exceptional case, that of dimension~2. Theorem~\ref{thm:exp} assumes
$\Dim{\H}\ge2$ while Theorem~\ref{thm:val} assumes $\Dim{\H}\ge3$. So
what about dimension 2?

Bell developed, in \cite{bell64} and \cite{bell66}, a hidden-variable theory for a two-dimensional Hilbert space $\H$. Here we summarize the improved version of Bell's theory due to Mermin \cite{mermin}, we simplify part of Mermin's argument, and we explain why the theory doesn't contradict Theorem~\ref{thm:exp}.

In the rest of this section, we work in the two-dimensional Hilbert
space $\H$. Let $\V$ be the set of value maps $v$ for all the
observables on $\H$.  In each pure state $\psi$, the hidden variables
should determine a particular member of $\V$. 

\begin{definition}\label{def:val}\rm
A \emph{value representation} for quantum systems described by $\H$ is a pair $(\Lambda,V)$ where
\begin{itemize}
\item $\Lambda$ is a probability space and
\item $V$ a function $\psi\to V_\psi$ on the pure states such that every $V_\psi$ is a map $\lambda\to V_\psi^\lambda$ from (the sample space of) $\Lambda$ onto $\V$.
\end{itemize}
Further, we require that, for any pure state $\psi$ and any observable $A$, the expectation $\int_\Lambda V_\psi^\lambda(A)\: d\lambda$ of the eigenvalue of $A$ agrees with the prediction $\bra{\psi}A\ket{\psi}$ of quantum theory:
\begin{equation}\label{bell1}
 \int_\Lambda V_\psi^\lambda(A)\: d\lambda = \bra{\psi}A\ket{\psi}
\end{equation}
\end{definition}

\medskip
Definition~\ref{def:val} is narrowly tailored for our goals in this section; in the full paper we will give a general definition of value representation. Notice that, if a random variable (in our case, the eigenvalue of $A$ in $\psi$) takes only two values, then the expected value determines the probability distribution. A priori we should be speaking about commuting operators and joint spectra but things trivialize in the 2-dimensional case. Recall Proposition~\ref{pro:jspec} and notice that, in the 2-dimensional Hilbert space, if operators $A,B$ commute, then one of them is a polynomial function of the other.

\begin{theorem}\label{thm:bell}
There exists a value representation for the quantum systems described by the two-dimensional Hilbert system $\H$.
\end{theorem}

\begin{proof}
Let $\vec\sigma$ be the triple of the Pauli matrices
$\displaystyle \sigma_x=
\begin{pmatrix}
  0&1\\1&0
\end{pmatrix}, \sigma_y=
\begin{pmatrix}
  0&-i\\i&0
\end{pmatrix}, \sigma_z=
\begin{pmatrix}
  1&0\\0&-1
\end{pmatrix}$.
For any unit vector $\vec n\in\R^3$, 
the dot product $\vec n\cdot\vec\sigma$ is a Hermitian operator with
eigenvalues $\pm1$. 
Every pure state of $\H$ is an eigenstate, for eigenvalue $+1$, of
$\vec n\cdot\vec\sigma$ for a unique $\vec n$.  We use the notation
\ket{\vec n} for this eigenstate. 

If \scr H represents the states of a spin-$\frac12$ particle, then the operator $\frac12\vec n\cdot\vec\sigma$ represents the spin component in the direction $\vec n$, and so \ket{\vec n} represents the state in which the spin is definitely aligned in the direction $\vec n$.  It is a special property of spin $\frac12$ that all pure states are of this form; for higher spins, a superposition of states with definite spin directions need not have a definite spin direction.

On $\H$, any Hermitian operator $A$ has the form $a_0I + \a\cdot\s$ for some scalar
$a_0\in\R$ and vector $\a\in\R^3$. The eigenvalues of $A$ are
$a_0\pm\nm{\a}$. Observables $a_0I+\a\cdot\s$ and $b_0I+\vec b\cdot\s$
commute if and only if $\a$ and $\vec b$ are either parallel or
antiparallel.

The desired probability space is the set $S^2$ of unit vectors in $\R^3$ with the uniform probability measure. Let $\m$ range over $S^2$. Then
\[
  V_{\n}^{\m} (a_0 I + \a\cdot\s) =
   \begin{cases}
   a_0 + \nm{\a}&\text{if }(\m+\n)\cdot\a \geq 0,\\
   a_0 - \nm{\a}&\text{if }(\m+\n)\cdot\a    < 0.
   \end{cases}
\]
It remains to check that
\begin{equation}\label{bell2}
\int_{S^2} V_\psi^{\m}(a_0 I + \a\cdot\s)\: d\m =
\bra{\n}(a_0 I + \a\cdot\s)\ket{\n}.
\end{equation}

We begin with a couple of simplifications.  First, we may assume that $a_0=0$, because a general $a_0$ would just be added to both sides of Equation~\eqref{bell2}.  Second, thanks to the rotational symmetry of the situation (where rotations are applied to all three of $\a$, $\n$ and $\m$), we may assume that the vector $\a$ points in the $z$-direction.  Finally, by scaling, we may assume that $\a=(0,0,1)$, so that the right side of Equation~\eqref{bell2} is $n_z$.

So our task is to prove that the average over $\m$ of the values assigned to $\sigma_z$ is $n_z$.  By definition, the value
assigned to $\sigma_z$ is $\pm1$, where the sign is chosen to agree
with that of $m_z+n_z$.  In view of how $\m$ is chosen, this
$m_z+n_z$ is the $z$-coordinate of a random point on the unit sphere
centered at $\n$.  So the question reduces to determining what
fraction of this sphere lies above the $x$-$y$ plane.

This plane cuts $S^2$ horizontally at a level $n_z$ below the sphere's
center.  By a theorem of Archimedes, when a sphere is cut
by a plane, its area is divided in the same ratio as the length of the
diameter perpendicular to the plane.   So the plane divides the sphere's area in
the ratio of $1+n_z$ (above the plane) to $1-n_z$ (below the plane).
That is, the value assigned to $\sigma_z$ is $+1$ with probability
$(1+n_z)/2$ and $-1$ with probability $(1-n_z)/2$.  Thus, the average
value of $\sigma_z$ is $n_z$, as required.
\end{proof}

Finally, we explain why Bell's theory doesn't contradict Theorem~\ref{thm:exp}. To obtain an expectation representation, we must extend the map $V$
convex-linearly to all density matrices. But no such extension exists.  Here is an example showing what goes wrong.
Consider the four pure states corresponding to spin in the directions
of the positive $x$, negative $x$, positive $z$ and negative $z$
axes.  The corresponding density operators are the projections
\[
\frac{I+\sigma_x}2,\quad\frac{I-\sigma_x}2,\quad
\frac{I+\sigma_z}2,\quad \frac{I-\sigma_z}2,
\]
respectively.  Averaging the first two with equal weights, we get
$\frac12 I$; averaging the last two gives the same result.  So a
convex-linear extension $T$ would have to assign to the density
operator $\frac12I$ the average of the probability measures assigned
to the pure states with spins in the $\pm x$ directions and also the
average of the probability measures assigned to pure states with spins
in the $\pm z$ directions.  But these two averages are visibly very
different.  The first is concentrated on the union of two unit spheres
tangent to the $y$-$z$-plane at the origin, while the second is
concentrated on the union of two unit spheres tangent to the
$x$-$y$-plane at the origin.


Thus, Bell's example of a hidden-variable theory for 2-dimensional
\scr H does not fit the assumptions in any of the expectation no-go
theorems.  It does not, therefore, clash with the fact that those
theorems, unlike the value no-go theorems, apply in the 2-dimensional
case.


\begin{thebibliography}{99}

\bibitem{bell64} John S. Bell, ``On the Einstein-Podolsky-Rosen   paradox,'' \emph{Physics} 1 (1964) 195--200.

\bibitem{bell66} John S. Bell, ``On the problem of hidden variables in   quantum mechanics,'' \emph{Reviews of Modern Physics} 38 (1966)   447--452.

\bibitem{spectral} Michael S. Birman and Michael Z. Solomjak,
  ``Spectral Theory of Self-Adjoint Operators in Hilbert Space,''
  Reidel 1987 (originally in Russian, Leningrad University Press
  1980).

\bibitem{G228} Andreas Blass and Yuri Gurevich, ``On Hidden Variables: Value and Expectation No-Go Theorems,'' arXiv:1509.06896.


\bibitem{bugajski} S\l awomir Bugajski, ``Classical frames for a   quantum theory---a bird's-eye view,'' \emph{International Journal of Theoretical Physics} 32 (1993) 969--977.

\bibitem{ferrieA} Christopher Ferrie and Joseph Emerson, ``Frame   representations of quantum mechanics and the necessity of negativity in quasi-probability representations,'' \emph{Journal of Physics A:
    Mathematical and Theoretical} 41 352001 (2008), also   arXiv:0711.2658.

\bibitem{ferrieB} Christopher Ferrie and Joseph Emerson, ``Framed Hilbert space: hanging the quasi-probability pictures of quantum theory,'' \emph{New Journal of Physics} 11 063040 (2009), also
  arXiv:0903.4843.

\bibitem{ferrieC} Christopher Ferrie, Ryan Morris and Joseph
  Emerson, ``Necessity of negativity in quantum theory,''
  \emph{Physical Review A} 82, 044103 (2010), also arXiv:0910.3198.

\bibitem{ks} Simon Kochen and Ernst Specker, ``The problem of hidden variables in quantum mechanics,'' \emph{Journal of Mathematics and Mechanics} 17 (1967) 59--87.

\bibitem{mermin} N. David Mermin, ``Hidden variables and the two theorems of John Bell,'' \emph{Reviews of Modern Physics} 65 (1993) 803--815.

\bibitem{schatten} Robert Schatten, ``Norm ideals of completely continuous operators,'' Springer Verlag 1970, 2nd edition.

\bibitem{spekkens} Robert W. Spekkens, ``Negativity and contextuality are equivalent notions of nonclassicality,'' \emph{Physics Review Lettters} 101(2) (2008) 020401, also arXiv:0710.5549.
%

\end{thebibliography}
\end{document}